\documentclass[letterpaper, 10 pt, conference]{ieeeconf}

\IEEEoverridecommandlockouts
\usepackage{cite}
\usepackage{amsmath, amssymb, amsfonts}
\usepackage{graphicx}
\usepackage{textcomp}
\usepackage{xcolor}
\usepackage{array}
\usepackage{float}
\usepackage{algorithm}
\usepackage{empheq}
\usepackage{algorithmic}
\usepackage{mathtools} 
\usepackage[hidelinks]{hyperref}

\newtheorem{Theo}{Theorem}
\newtheorem{Lem}{Lemma}

\newtheorem{Rem}{Remark}

\newtheorem{Def}{Definition}

\newtheorem{Prop}{Proposition}

\newcolumntype{C}[1]{>{\centering\arraybackslash}m{#1}}

\def\BibTeX{{\rm B\kern-.05em{\sc i\kern-.025em b}\kern-.08em
    T\kern-.1667em\lower.7ex\hbox{E}\kern-.125emX}}

\title{\LARGE \bf Consensus and Synchronization of Multi-agent Systems over Finite Fields - Graph Topologies}

\author{Kristian Hengster-Movrić, Šimon Lehký and Farnaz Adib Yaghmaie%
\thanks{The work of Kristian Hengster-Movrić and Šimon Lehký was co-funded by the European Union under the project ROBOPROX (reg. no. CZ.02.01.01/00/22\_008/0004590) and supported by the Grant Agency of the Czech Technical University in Prague (grant No. SGS25/144/OHK3/3T/13). Farnaz Adib Yaghmaie is supported by the Excellence Center at Linköping–Lund in Information Technology (ELLIIT), ZENITH, and partially by Sensor informatics and Decision-making for the Digital Transformation (SEDDIT). This work was performed within the Competence Center SEDDIT (Sensor Informatics and Decision making for the DIgital Transformation), supported by Sweden’s Innovation Agency Vinnova within the research and innovation program Advanced digitalization.}%
\thanks{Kristian Hengster-Movrić and Šimon Lehký are with Faculty of Electrical Engineering, Czech Technical University, Prague, Czech Republic
        {\tt\small \{hengskri, lehkysim\}@fel.cvut.cz}}%
\thanks{Farnaz Adib Yaghmaie is with Faculty of Electrical Engineering, Linköping University, Linköping, Sweden
        {\tt\small farnaz.adib.yaghmaie@liu.se}}%
}

    \usepackage{ifthen}
\usepackage{eso-pic}
\usepackage{xcolor}

\newif\ifarxivstamp
\arxivstamptrue

\newcommand{\ArxivStampVenue}{IEEE Conference on Decision and Control (Regular Paper)}
\newcommand{\ArxivStampDate}{September 18, 2026}
\newcommand{\ArxivStampFunding}{Co-funded by the European Union under the project ROBOPROX (reg. no. CZ.02.01.01/00/22\_008/0004590).}

\newcommand{\ArxivSubmittedStampOverlay}{%
    \ifarxivstamp
    \AddToShipoutPictureFG*{%
    \AtPageUpperLeft{%
    \raisebox{-1.6\baselineskip}[0pt][0pt]{%
        \begin{minipage}{\paperwidth}\centering\footnotesize
            \color{black}
            \hrule \vspace{2pt}
            \textbf{Manuscript accepted for \ArxivStampVenue, \ArxivStampDate.} \\
            Preprint (final version). \ArxivStampFunding
            \vspace{2pt}\hrule
            \end{minipage}%
            }%
        }%
    }%
    \fi
}

\begin{document}
\ArxivSubmittedStampOverlay

\maketitle
\thispagestyle{empty}
\pagestyle{empty}

\begin{abstract}
    This paper presents cooperative protocols for multi-agent systems with agents having a finite state-space. Both scalar single-integrator consensus and general LTI system synchronization are considered. Systems having a finite state-space describe agents with minimal memory capacity processing only a finite alphabet. Such systems are remarkably resilient to communication noise. The crucial problem, however, is to construct the admissible communication topology, which is NP-hard. We address this by efficiently exploring the subsets of admissible graph matrices and propose two new algorithms to generate them. Simulations validate the proposed approach. 
\end{abstract}

\section{Introduction and Preliminaries} \label{sec:intro} \noindent
Multi-agent systems consisting of agents with minimal memory capacity, processing only values from a finite alphabet, are becoming increasingly interesting due to emerging IoT requirements~\cite{chen2017programmable_galois}, capacity and memory-constrained sensor networks, and for describing discretized orientations of rigid bodies~\cite{pasqualetti2014consensus_networks}, to name just a few compelling applications.

Consensus of real scalar single-integrator agents over nonnegative real edge-weight graphs has been long studied from multiple perspectives. Fundamental results on the graph topology necessary for consensus of single-integrators were extended to synchronization of general linear time-invariant (LTI) agents over the same underlying graphs. A spanning tree subgraph is found to be the minimal necessary requirement for cooperative synchronization of dynamical systems evolving in continuous-time and discrete-time, on finite-dimensional vector spaces over continuum real or complex number fields. Extensions of consensus and synchronization protocols to \textit{signed} real edge-weights~\cite{altafini2013consensus_antagonistic,zhang2014bipartite} and \textit{complex} edge-weight graphs~\cite{lou2015distributed_surrounding, dong2014complex}, both for single-integrator and general LTI agents followed thereafter~\cite{yaghmaie2017multiparty_consensus, yaghmaie2017bipartite_cooperative}. 

Systems over finite fields and finite-state automata~\cite{booth1967sequential}, have likewise been studied for a long time~\cite{elspas1959theory_autonomous, friedland1959LinearMS,hartmanis1959LinearMS}. Theory exists for LTI sequential finite-field systems evolving in discrete-time on finite-dimensional vector spaces over finite (\textit{Galois}) fields, \textit{Hamming spaces}. Frequency domain $Z$-transformation transfer function analysis, formally the same as for conventional discrete-time systems over fields of real and complex numbers, is useful there as well~\cite{friedland1959LinearMS, hartmanis1959LinearMS, booth1967sequential}, as are the polynomial-based approaches~\cite{friedland1959LinearMS, hartmanis1959LinearMS}. Nevertheless, cooperative synchronization of such systems over graphs with finite-field edge-weights has not been hitherto systematically investigated in a satisfactorily unified framework, although~\cite{pasqualetti2014consensus_networks} makes crucial first steps in that direction, while~\cite{xu2014leaderfollowing,meng2020synchronization} propose extensions to more general agent dynamics and~\cite{li2019leader_follower} even addresses time-delays. For scalar agents, \cite{lin2022synchronization} considers randomly switching graphs. For vector agents with specific dynamics, \cite{lin2024leader} investigates leader-following finite-field networks, in a centralized way. 

This paper develops an analysis and design methodology for consensus of identical finite-field LTI systems, following the single-agent-based synchronizing region approach,~\cite{lewis2014cooperative}. For an LTI system that is finite-field stabilizable, we show that a stabilizing controller gain can be obtained directly and transparently from either the controllable canonical form or the Kalman decomposition of the single-agent dynamics. Thus, the essential challenge is \textit{not} the controller design itself, but rather the identification of admissible communication structures, (graph topologies and edge-weights), that render a cooperatively stabilizing control feasible. The finite-field admissible graph topology problem is identified in~\cite{pasqualetti2014consensus_networks} as NP-hard. In this work, we demonstrate that the design of admissible communication structures for synchronization is completely independent of the single-agent characteristics. This decoupling allows a separation of the graph design from single-agent controller synthesis.

The contributions of this paper are twofold. First, we develop a unified analysis and design framework for state consensus and synchronization of finite-field multi-agent systems. In contrast to~\cite{pasqualetti2014consensus_networks}, which considers only single-integrator dynamics, we address general identical LTI agents. This also generalizes the results of~\cite{meng2020synchronization}, which considers single-agent dynamics of a special form. In comparison to \cite{lin2024leader}, our approach is based on single-agent characteristics. Moreover, in contrast to~\cite{xu2014leaderfollowing}, we do not require directed acyclic graphs (DAG), nor do we require the equality of all in-degrees, but allow for all finite-field admissible graph topologies. Our analysis is based directly on the properties of the communication graph, rather than its derived \textit{state-transition graph}~\cite{pasqualetti2014consensus_networks}. Unlike \cite{lin2022synchronization}, which considers randomly switching graphs, we assume known constant graphs, but with general LTI agents instead of first-order scalar ones. Second, we propose efficient algorithms for generating admissible graphs that guarantee finite-field consensus, and we analyze their computational complexity. The proposed algorithms rely on matrix similarity transformations that preserve beneficial properties of the graph adjacency matrix, rather than on generating graph adjacency matrices directly. Unlike the sampling-and-rejection construction of~\cite{pasqualetti2014consensus_networks}, which draws random matrices and discards singular ones at the cost of $\mathcal{O}(N^3)$ \textit{per} trial, our triangular-form construction guarantees an admissible similarity transformation by design, reducing the complexity to $\mathcal{O}(N^2)$ with no determinant computations. To our knowledge, these results have not been previously reported.

The paper is organized as follows: Section~\ref{sec:intro} presents the introduction and preliminaries on finite fields, graphs and notation. Section~\ref{sec:scalar_single_int} introduces the scalar single-integrator consensus, while Section~\ref{sec:general_lti_seq} considers general finite-field LTI agents. Section~\ref{sec:math_foundation} elucidates the structures to be explored; Section~\ref{sec:algorithms} presents the efficient search algorithms; Section~\ref{sec:sim_results} gives a numerical example; Section~\ref{sec:conclusion} concludes the paper.

\subsection{\texorpdfstring{Finite Field $\mathbb{F}_{p}$}{Finite Field F\_p}} \label{subsec:fin_field} \noindent
The finite field $\mathbb{F}_{p}$ consists of the set of integers $\{0, 1, \dots, p-1\}$, with addition and multiplication performed modulo $p$. Throughout this paper, $p$ is a \textit{prime} number, so that $\mathbb{F}_{p}$ forms a \textit{field}. All algebraic operations in this paper are performed modulo $p$.

For a matrix $T \in \mathbb{F}_{p}^{m \times n}$, the entry in the $i$-th row and $j$-th column is written as $T_{ij} \in \mathbb{F}_{p}$, or $T_{i,j}$ where needed for clarity. The whole $i$-th row ($j$-th column) is denoted by $[T]_{i,:}$ ($[T]_{:,j}$). The transpose of $T$ is denoted by $T^{\top}$. Let $T \in \mathbb{F}_{p}^{s \times s}$ be an $s \times s$ matrix. The \emph{trace} of $T$ is defined as $\operatorname{Tr}(T) = \sum\nolimits_{i=1}^s T_{ii} \in \mathbb{F}_{p},$ and the \emph{determinant} of $T$ is denoted by $\det(T)$, where all additions and multiplications are performed modulo $p$. Equivalently, $T$ is invertible over $\mathbb{F}_{p}$ \textit{if and only if} $\det(T) \neq 0$ in $\mathbb{F}_{p}$. For an invertible $T$ and any matrix $A$, $T^{-1}AT$ is \textit{similar} to $A$, $T^{-1}AT\sim A$. Matrix similarity with a permutation matrix $P$ is termed \textit{cogredience}, $P^{\top}AP\simeq A$. $I_N$ denotes an $N\times N$ identity matrix. A vector of ones with appropriate dimension is denoted by $\mathbf{1} = (1, 1, \dots, 1)^{\top}$ or $\mathbf{1}_N$ when dimension is of concern. Similarly, $\mathbf{0} = (0, 0, \dots, 0)^{\top}$ or $\mathbf{0}_N$ denote a vector of zeros. Let $\mathbf{e}_i$ denote the $i$-th basis vector of $\mathbb{F}_{p}^{n}$ whose entries are all $0$ except for the $i$-th entry, which is $1$.

The \textit{characteristic polynomial} of a matrix $A \in \mathbb{F}_p^{n \times n}$, (where $\lambda$ is an \textit{indeterminate} over $\mathbb{F}_{p}$), is defined as
\begin{equation*}
    P_A(\lambda) \coloneqq \det(\lambda I_n - A) = \Pi_i(\lambda - \lambda_i)\Pi_j\phi_j(\lambda),
\end{equation*}
so the eigenvalues $\lambda$ satisfy $P_A(\lambda) = 0$. In the above, the $\phi_j(\lambda)$ terms are higher-order polynomials \textit{irreducible} over $\mathbb{F}_p$.

\begin{Rem}
    Finite fields not being algebraically closed has crucial implications for the matrix eigenvalue problem.
\end{Rem}

Since for any positive integer $d$ there exist irreducible polynomials of degree $d$ over $\mathbb{F}_p$, a finite-field matrix $A \in \mathbb{F}_{p}^{n \times n}$ \textit{does not} necessarily have $n$ eigenvalues in $\mathbb{F}_p$, if $P_A(\lambda)$ has irreducible factors of degree $d>1$ over $\mathbb{F}_p$. If $P_A(\lambda)$ splits completely into linear factors over $\mathbb{F}_{p}$, the matrix $A \in \mathbb{F}_{p}^{n \times n}$ can be rendered in a Jordan canonical form \textit{via} matrix similarity transformation. Such a matrix $A \in \mathbb{F}_{p}^{n \times n}$ then has a \textit{full} set of $n$ eigenvalues in $\mathbb{F}_p$. 

\begin{Rem}
    A matrix $E \in \mathbb{F}_p^{n \times n}$ is \textit{nilpotent} if $E$ has $n$ eigenvalues in $\mathbb{F}_p$, all equal to zero; $P_E(\lambda)=\lambda^n$. Equivalently, $E^n = 0$ by the Cayley--Hamilton theorem. Such matrices correspond to finite-time stable LTI systems.
\end{Rem}

\begin{Def}
    A matrix pair $(A,B)$ is \textit{finite-field stabilizable} if its uncontrollable part is \textit{finite-time stable}, with the \textit{nilpotent} pertaining system matrix block.
\end{Def}

\begin{Lem}
    \label{Lem:gain}
    A finite-field stabilizable pair $(A,B)$ admits a stabilizing feedback gain matrix $K$ that renders the closed-loop system matrix $A-BK$ nilpotent. For a single-input case, this stabilizing feedback gain is \textit{unique}. For a multi-input case, however, a finite number of possibilities exist.  
\end{Lem}

\begin{proof}
    In the controllable canonical form, for $A-BK$ to be nilpotent, $P_{A-BK}(\lambda)=\lambda^n$, $K$ must cancel all coefficients of the characteristic polynomial $P_A(\lambda)$ in $A$. In a single-input case, this leaves but one option. For a multi-input case, the problem is underdetermined, allowing for several solutions. In a finite-field context, the number of these solutions is finite. A detailed proof is omitted for brevity.
\end{proof}

\subsection{Graphs with Finite-field Edge-weights} \label{subsec:graphs_fin_fields_weights} \noindent
For a graph $\mathcal{G}(\mathcal{V}, \mathcal{E}, E)$, having $|\mathcal{V}|=N$ nodes and $|\mathcal{E}|=M$ edges, with finite-field edge-weights, the \textit{adjacency matrix} $E \in \mathbb{F}_{p}^{N \times N}$, with elements $e_{ij} \in \mathbb{F}_{p}$, is an $N \times N$ matrix; $e_{ij} \neq 0$ if $(i,j) \in \mathcal{E}$ and $e_{ij} = 0$ if $(i,j) \notin \mathcal{E}$. Conversely, any $N \times N$ finite-field matrix $E$ induces a graph $\mathcal{G}(\mathcal{V}, \mathcal{E}, E)$ with finite-field edge-weights. If there is a node emanating directed paths to all other nodes in the graph, the graph contains a \textit{spanning tree} subgraph, and each such globally influential node is a possible \textit{root} of a spanning tree.

\section{Scalar Single-Integrator Agent Consensus} \label{sec:scalar_single_int} \noindent
Consider $N$ agents with state dynamics
\begin{equation}
    \label{eq:scalar_consensus}
    x_i(k+1) = \sum\nolimits_j e_{ij} x_j(k), \quad x_i\in \mathbb{F}_{p}, ~~i = 1, \dots, N,
\end{equation}
where possibly $e_{ii} \neq 0$.

In total form with $x \coloneqq \begin{bmatrix} x_1 & \dots & x_N \end{bmatrix}^{\top} \in \mathbb{F}_{p}^{N}$,~\eqref{eq:scalar_consensus} reads
\begin{equation}
    \label{eq:x_Ex}
    x(k+1) = Ex(k) \in \mathbb{F}_{p}^N, \quad E \in \mathbb{F}_{p}^{N \times N}.
\end{equation}

\begin{Theo}
    \label{Theo:Theorem1}
    \textup{\cite{pasqualetti2014consensus_networks}}. Let $E$ be either a) nilpotent with the characteristic polynomial $P_E(\lambda) = \lambda^N$, or b) row-stochastic, $E \mathbf{1}_N = \mathbf{1}_N$, with the characteristic polynomial $P_E(\lambda) = (\lambda - 1)\lambda^{N-1}$ and left eigenvector $p^{\top} \in \mathbb{F}_{p}^N, ~p^{\top} E = p^{\top}$ pertaining to the simple eigenvalue $1$, unique up to scalar multiplication, satisfying $p^{\top} \mathbf{1}_N \neq 0$. Then and only then does the state of the finite-field LTI system~\eqref{eq:x_Ex} converge to consensus, $x(k) \rightarrow \alpha \mathbf{1}_N$, for all initial conditions in finitely many discrete iterations. Nilpotent $E$ leads to $\alpha =0$, while a row-stochastic $E$ with spectrum $\sigma(E)=(1,0,\dots,0)$ leads to nontrivial final consensus value $\alpha = p^{\top} x(0) / p^{\top} \mathbf{1}_N$.
\end{Theo}

\begin{proof}
    Sufficiency follows from modal decomposition~\cite[Theorem~4.4]{pasqualetti2014consensus_networks}. A nilpotent matrix $E$ leads to zero final consensus value for all initial conditions. For a row-stochastic matrix $E$ with spectrum $\sigma(E)=(1,0,\dots,0)$, there exists the left eigenvector $p^\top$ corresponding to the simple eigenvalue $1$, unique up to scalar multiplication. The final consensus value follows from the conserved quantity $p^{\top} E = p^{\top}$, $p^{\top} x(k) = \textit{const}$; $p^{\top} x(k) = p^{\top} x(0)$, so that if $x(k) \rightarrow \alpha \mathbf{1}_N$, the initial and final states are related as
    \begin{equation*}
        p^{\top} \alpha \mathbf{1}_N = p^{\top} x(0) \Rightarrow \alpha = p^{\top} x(0) / p^{\top} \mathbf{1}_N,
    \end{equation*}
    assuming $p^{\top} \mathbf{1}_N \neq 0$. That $p^{\top} \mathbf{1}_N \neq 0$ ensures the final consensus value is well-defined for all initial conditions. If $p^{\top} \mathbf{1}_N = 0$, consensus is possible only if $p^{\top} x(0) = 0$, hence \textit{not} for all initial conditions.
    
    Necessity is proven as follows. If~\eqref{eq:x_Ex} converges to consensus in finitely many steps for all initial conditions, the final consensus value is either always zero, or sometimes nonzero. In the first case, the matrix $E$ must be nilpotent, with $E^k = 0$ for some $k>0$. For nonzero consensus value, the matrix $E$ must be row-stochastic, $E\mathbf{1}_N = \mathbf{1}_N$, with $E^k=\mathbf{1}_N p^{\top}/p^{\top}\mathbf{1}_N$ for some $k > 0$, where $p^{\top}E = p^{\top}$. As $E\mathbf{1}_N = 1\cdot\mathbf{1}_N$, $E$ is similar to a triangular matrix $J$; $E = VJW$, $W = V^{-1}$, with the first column of $V$ being $\mathbf{1}_N$, the first row of $W$ being the normalized $p^\top/p^\top\mathbf{1}_N$, and
    \begin{equation*}
        J = \begin{bmatrix} 
            1 & \mathbf{0}^\top \\ 
            \mathbf{0} & \tilde{J}~~  
        \end{bmatrix}.
    \end{equation*}
    As $E^k = VJ^kW$, if $E^k = \mathbf{1}_Np^{\top}/p^{\top}\mathbf{1}_N$, it must hold that 
    \begin{equation*}
        J^k = \begin{bmatrix} 
            1 & \mathbf{0}^\top \\ 
            \mathbf{0} & \tilde{J}^k 
        \end{bmatrix} = 
        \begin{bmatrix} 
            1 & \mathbf{0}^\top \\ 
            \mathbf{0} & 0~~  
        \end{bmatrix}.
    \end{equation*}
    Therefore, $\tilde{J}^k = 0$, so $\tilde{J}$ must be nilpotent; $P_{\tilde{J}}(\lambda) = \lambda^{N-1}$, and $\tilde{J}$ admits the Jordan form
    \begin{equation*}
        \tilde{J} \sim \begin{bmatrix} 
            0 & * & \dots & 0 \\ 
            0 & 0 & \ddots & \vdots \\ 
            \vdots & \vdots & \ddots & * \\ 
            0 & 0 & \dots & 0 
        \end{bmatrix},
    \end{equation*}
    whence, $P_J(\lambda) = (\lambda-1)P_{\tilde{J}}(\lambda) = (\lambda-1)\lambda^{N-1} = P_E(\lambda)$.

    Moreover, a row-stochastic $E$ has an invariant subspace $\operatorname{span}(\mathbf{1}_N)$---\textit{the consensus subspace}. The restriction of the matrix $E$ to the complement of $\operatorname{span}(\mathbf{1}_N)$ is given by $\tilde{J}$.
    
    Define the synchronization error vector in total form as
    \begin{equation*}
        \varepsilon\coloneqq\left(I_N - \frac{\mathbf{1}_Np^{\top}}{p^{\top}\mathbf{1}_N}\right)x ~~ \in\mathbb{F}_p^{N}.
    \end{equation*}
    
    By construction $p^{\top}\varepsilon = 0$, so $\varepsilon$ is restricted to an $(N-1)$-dimensional hyperplane of $\mathbb{F}_p^N$ having no component along $\operatorname{span}(\mathbf1_N)$; $p^{\top}\varepsilon = 0 ~\land ~\varepsilon = \alpha \mathbf{1}_N \implies \alpha = 0, ~(p^{\top} \mathbf{1}_N \neq0)$, hence $\varepsilon = 0$.
    
    As $\varepsilon$ also satisfies~\eqref{eq:x_Ex}, $\varepsilon(k+1) = E\varepsilon(k)$, restricted to the complement of $\operatorname{span}(\mathbf{1}_N)$, the $\varepsilon$-dynamics is governed by $\tilde{J}$, converging to zero in a finite number of iterations upper bounded by the \textit{nilpotency index} of $\tilde{J}$.
\end{proof}

\begin{Rem} 
    Finite-field convergence proceeds necessarily in a finite number of steps, making finite-field consensus crucially different from consensus on $\mathbb{R}$ and $\mathbb{C}$, where stability is asymptotic. A nilpotent $E$ leads to trivial consensus for all initial conditions, in finitely many steps~\cite{pasqualetti2014consensus_networks}, which is of no interest here. A row-stochastic $E$ leads to a final consensus value $\alpha$ that depends on the initial conditions.
\end{Rem}

\begin{Def}
    Given a finite-field edge-weight graph with a row-stochastic adjacency matrix $E$, $E \mathbf{1}_N = \mathbf{1}_N$, the finite-field graph \textit{Laplacian matrix} is defined as $L \coloneqq I_N - E$, having all row sums equal to zero, $L \mathbf{1}_N = \mathbf{0}_N$.
\end{Def}

For the graph matrix $E$ considered here, this finite-field graph Laplacian matrix has a full set of $N$ eigenvalues in $\mathbb{F}_p$ as well; $P_L(\lambda) = \lambda(\lambda - 1)^{N-1}$.

\section{General LTI Agent Synchronization} \label{sec:general_lti_seq} \noindent
Consider $N$ identical LTI agents with state dynamics
\begin{equation}
    \begin{aligned}
        x_i(k+1) = Ax_i(k) + Bu_i(k), \qquad i = 1, \dots, N, \\
        x_i \in \mathbb{F}_{p}^n, u_i \in \mathbb{F}_{p}^m, A \in \mathbb{F}_{p}^{n \times n}, B \in \mathbb{F}_{p}^{n \times m}.
    \end{aligned}
\end{equation}
The control input of each agent is chosen in terms of the finite-field graph Laplacian matrix, $L = I_N - E$, as 
\begin{equation}
    u_i(k) = -K \sum\nolimits_j L_{ij} x_j(k) = -K \Big(x_i(k) - \sum\nolimits_j e_{ij} x_j(k)\Big),
\end{equation}
where $K \in \mathbb{F}_{p}^{m \times n}$ is the cooperative feedback gain matrix, leading to the closed-loop single-agent system 
\begin{equation} 
    \label{eq:closed_loop_single_agent}
    x_i(k+1) = Ax_i(k) - BK \sum\nolimits_j L_{ij} x_j(k).
\end{equation}
In total form with $x \coloneqq \begin{bmatrix} x_1^{\top} & \dots & x_N^{\top} \end{bmatrix}^{\top} \in \mathbb{F}_{p}^{Nn}$,~\eqref{eq:closed_loop_single_agent} reads
\begin{equation}
    \label{eq:closed_loop_global}
    x(k+1) = \left[I_N \otimes A - L \otimes BK\right]x(k)
\end{equation}

\begin{Theo}
    \label{Theo:Theorem2}
    Let $E$ be row-stochastic, $E \mathbf{1}_N = \mathbf{1}_N$, with the characteristic polynomial $P_E(\lambda) = (\lambda - 1)\lambda^{N-1}$ and the left eigenvector $p^{\top} \in \mathbb{F}_{p}^N$, $p^{\top} E = p^{\top}$, pertaining to the simple eigenvalue $1$, satisfying $p^{\top} \mathbf{1}_N \neq 0$. Let the pair $(A, B)$ be finite-field \textit{stabilizable} with a stabilizing feedback gain matrix $K$. Then the multi-agent system~\eqref{eq:closed_loop_single_agent} synchronizes in finitely many discrete iterations to an autonomous behavior of the matrix $A$; $x_i(k) \rightarrow \alpha(k), \forall i$, where $\alpha(k) = \sum\nolimits_i p_i x_i(k) / \sum\nolimits_i p_i$, and $\alpha(k+1) = A \alpha(k)$.
\end{Theo}

\begin{proof}
    Given the conditions on $E$, the pertaining finite-field graph Laplacian matrix $L$ has a full set of $N$ eigenvalues; hence, a nonsingular matrix $T$ transforms the graph Laplacian matrix through matrix similarity transformation into a \textit{triangular} matrix $\Lambda$ with the eigenvalues of the Laplacian matrix on its diagonal, 
    \begin{equation}
        T^{-1}LT = \Lambda \Rightarrow T^{-1}(I_N - E)T = (I_N - T^{-1}ET) = \Lambda.
    \end{equation}

    This allows a structured matrix similarity transformation of the system matrix in~\eqref{eq:closed_loop_global} as
    \begin{equation}
        \begin{aligned}
            (T^{-1} \otimes I_n)\left[I_N \otimes A - L \otimes BK\right](T \otimes I_n) &= \\
            \hfill = I_N \otimes A &- \Lambda \otimes BK.
        \end{aligned}
    \end{equation}
    The diagonal entries of $\Lambda$ are eigenvalues of $L$, namely one zero and $N-1$ ones. Hence $\Lambda_{ii} \in \{ 0,1 \}, \forall i$, with exactly one index having $\Lambda_{ii}=0$ and the rest having $\Lambda_{ii}=1$. The stability of the whole transformed system matrix 
    \begin{equation}
        \label{closedloopblckdiag}
        I_N \otimes A - \Lambda \otimes BK = \begin{bmatrix} 
                A & * & \dots & * \\ 
                0 & A - BK & \dots & * \\ 
                0 & 0 & \dots & * \\ 
                \vdots & \vdots & \ddots & \vdots \\ 
                0 & 0 & \dots & A - BK 
            \end{bmatrix}
    \end{equation}
    is determined by the diagonal blocks $A - \Lambda_{ii} BK$, which take only two forms: $A, ~\text{if } \Lambda_{ii} = 0$, and $A - BK, ~\text{if } \Lambda_{ii} = 1$. Choosing a $T$ yielding the Jordan canonical form for $\Lambda$, the simple zero eigenvalue of the Laplacian matrix $L$ leads to the diagonal block $A$. This then reveals the invariant subspace where the restriction of the total dynamics~\eqref{eq:closed_loop_global} is completely described by the system matrix $A$---\textit{the synchronization subspace}, corresponding to the synchronized behavior determined by the uncontrolled agent dynamics. The only other diagonal block, $A - BK$, determines the error from a synchronized state. 
    
    Namely, define the synchronization error in total form as
    \begin{equation*}
        \varepsilon \coloneqq \left[\left(I_N-\frac{\mathbf{1}_Np^{\top}}{p^{\top}\mathbf{1}_N}\right) \otimes I_n\right]x ~~\in\mathbb{F}_p^{Nn}.
    \end{equation*}
    By construction, $\sum_{i=1}^{N} p_i \varepsilon_i = (p^\top \otimes I_n)\varepsilon = 0$, so $\varepsilon$ is restricted to an $(N-1)n$-dimensional subspace of $\mathbb{F}_p^{Nn}$, having no components along the synchronization subspace, $(p^\top \otimes I_n)\varepsilon = 0 ~\land ~\varepsilon = \mathbf{1}_N\otimes \alpha \implies \varepsilon = 0,  ~~(p^{\top} \mathbf{1}_N \neq0)$.
    
    As $\varepsilon$ also satisfies dynamics~\eqref{eq:closed_loop_global},
    \begin{equation*}
        \varepsilon(k+1) = \left[I_N \otimes A - L \otimes BK\right]\varepsilon(k),
    \end{equation*}
    restricted to the complement of the invariant synchronization subspace, $\operatorname{span}(\mathbf{1}_N \otimes \alpha)$, the $\varepsilon$-dynamics is governed by the triangular submatrix in \eqref{closedloopblckdiag} with diagonal blocks $A - BK$.
    
    For synchronization of system~\eqref{eq:closed_loop_global}, the matrix $A - BK$ must thus be finite-time stable, nilpotent; $P_{A - BK}(\lambda) = \lambda^n$. This is guaranteed for a \textit{finite-field stabilizable} pair $(A,B)$, by a finite-field stabilizing feedback gain $K$, (unique for single-input systems; \textit{cf}. Lemma~\ref{Lem:gain}). This way, the common modes are \textit{not} stabilized to 0, but the error from a synchronized state becomes 0 in a finite number of steps, upper bounded by the nilpotency index of the block-triangular submatrix of \eqref{closedloopblckdiag} with diagonal blocks $A - BK$.

    Moreover, consider the linear combination of agents' states 
    \begin{equation*}
        \alpha(k) \coloneqq \sum\nolimits_i p_i x_i(k) / \sum\nolimits_i p_i,
    \end{equation*}    
    with elements $p_i$ of the left eigenvector of the matrix $E$ pertaining to the simple eigenvalue $1$, satisfying $p^{\top} \mathbf{1}_N \neq 0$. One has that $p^{\top} E = p^{\top} \Rightarrow p^{\top} L = \mathbf{0}_{N}^{\top}$, because
    \begin{equation}
        p^{\top} L =p^{\top}(I_N - E) = p^{\top} - p^{\top} E = p^{\top} - p^{\top} = \mathbf{0}_{N}^{\top},
    \end{equation}
    so $\sum\nolimits_i p_i L_{ij} = 0$, $\forall j$. Therefore, in dynamics of this linear combination of states, the network contributions cancel out, 
    \begin{equation}
        \begin{aligned}
            \sum\nolimits_i p_i x_i(k{+}1) &= \sum\nolimits_i p_i Ax_i(k) - BK \sum\nolimits_{i,j} p_i L_{ij} x_j(k) \\
            &= \sum\nolimits_i p_i Ax_i(k) = A \sum\nolimits_i p_i x_i(k).
        \end{aligned}
    \end{equation}
    Hence, if the agents' states synchronize, $x_i(k) \rightarrow \alpha(k), \forall i$, one has that $\sum\nolimits_i p_i x_i(k) \rightarrow \sum\nolimits_i p_i \alpha(k)$, which means $\sum\nolimits_i p_i x_i(k) / \sum\nolimits_i p_i \rightarrow \alpha(k)$. Therefore, $\alpha(k+1) = A\alpha(k)$, with $\alpha(0) \coloneqq \sum\nolimits_i p_i x_i(0) / \sum\nolimits_i p_i$.
\end{proof}

\begin{Rem}
    Given a finite-field \textit{stabilizable} $(A,B)$ pair, one can use the controllable canonical form, or Kalman decomposition of the system $(A,B)$, to find a matrix $K$ rendering the closed-loop system matrix $(A-BK)$ nilpotent; $P_{A - BK}(\lambda) = \lambda^{n}$. \textbf{The crucial problem is thus to find the graph topologies and finite-field edge-weights that yield the required adjacency matrix $E$ and its pertaining graph Laplacian matrix $L$.} Finding all such graphs is NP-hard~\cite{pasqualetti2014consensus_networks}; nevertheless, this is independent of the single-agent dynamics $(A,B)$, allowing for a design separation formally akin to the conventional \textit{synchronizing region} approach,~\cite{lewis2014cooperative}. Our results thus crucially differ from those in \cite{xu2014leaderfollowing}, where the cooperative feedback $K$ depends on the graph topology. Moreover, note that here, due to the spectrum required of $E$, and $L$, the synchronizing region reduces to a single point, $1$.
\end{Rem}

One can find all matrices $E$ that satisfy the required conditions by an exhaustive search, which is a combinatorially large NP-hard problem,~\cite{pasqualetti2014consensus_networks}. To obviate listing all the possibilities, recursive constructive schemes could build a subclass of admissible graph matrices, which is not perfectly exhaustive, but practically useful. One way is the Kronecker product graph construction, mentioned in~\cite[Theorem~5.2]{pasqualetti2014consensus_networks}. Another possibility is afforded by a special class of matrix similarity transformations, as detailed in the following result.

\begin{Prop} 
    \label{Prop:mat_sim_transformation}
    A matrix similarity transformation with a row-stochastic nonsingular matrix $T$ preserves the crucial properties of the graph matrix $E$ required for consensus.
\end{Prop}

\begin{proof}
    Given an admissible matrix $E$, $T^{-1}ET = \tilde{E}$ has the \textit{same spectrum} as $E$, $\sigma(T^{-1}ET) = \sigma(E) = (1, 0, \dots, 0)$, by matrix similarity relation. If  $T \mathbf{1}_N = \mathbf{1}_N$, then $\tilde{E} \mathbf{1}_N = T^{-1}ET \mathbf{1}_N = T^{-1}E \mathbf{1}_N = T^{-1} \mathbf{1}_N = \mathbf{1}_N$, preserving the \textit{row-stochastic} property of $E$ as well. Those are sufficient requirements for consensus. Note, namely, that if an invertible matrix $T$ is row-stochastic, $T \mathbf{1}_N = \mathbf{1}_N$, then so is its inverse
    \begin{equation}
        T \mathbf{1}_N = \mathbf{1}_N \Rightarrow T^{-1}T \mathbf{1}_N = T^{-1} \mathbf{1}_N \Rightarrow T^{-1} \mathbf{1}_N = \mathbf{1}_N.
    \end{equation}
    
    Moreover, the left eigenvector of $E$ pertaining to the simple eigenvalue $1$, $p^{\top} E = p^{\top}$, transforms under this matrix similarity transformation as $\tilde{p}^{\top} T^{-1} = p^{\top}$,
    \begin{equation}
        \tilde{p}^{\top} \tilde{E} = \tilde{p}^{\top} \Rightarrow \tilde{p}^{\top} T^{-1}ET = \tilde{p}^{\top} \Rightarrow \tilde{p}^{\top} T^{-1}E = \tilde{p}^{\top} T^{-1},
    \end{equation}
    whence $p^{\top} \mathbf{1}_N = \tilde{p}^{\top} T^{-1} \mathbf{1}_N = \tilde{p}^{\top} \mathbf{1}_N$, so that if $p^{\top} \mathbf{1}_N \neq 0$ one also has that $\tilde{p}^{\top} \mathbf{1}_N \neq 0$, as required.
\end{proof}

Proposition~\ref{Prop:mat_sim_transformation} motivates a search for all \textit{nonsingular} row-stochastic matrices $T$ over $\mathbb{F}_p$, which is likewise a combinatorially large problem.

\begin{Rem}
    Cogredience transformation with a permutation matrix $P$, $P^{-1} = P^{\top}$, certainly satisfies conditions of Proposition~\ref{Prop:mat_sim_transformation} as $P \mathbf{1}_N = \mathbf{1}_N$ over a general finite field, but that transformation produces equivalent, \textit{isomorphic}, graphs. There exist nonisomorphic graphs with the \textit{same} spectrum (\textit{cospectral} or \textit{isospectral} graphs), and matrix similarity transformations with row-stochastic matrices $T$, up to permutations, are one way of finding those. 
\end{Rem}

\begin{Rem}
    \label{Rem:TP_row_stochastic}
    Given a row-stochastic invertible matrix $T$ and a permutation matrix $P$,  $TP$ is row-stochastic and invertible as well. The entire permutation matrix subgroup thus defines \textit{left cosets} for each $T$. For an admissible graph matrix $E$, matrix similarity transformations with $T$ and $TP$ for any permutation matrix $P$ give \textit{isomorphic} graphs; $(TP)^{-1}E(TP) = P^{\top}(T^{-1}ET)P\simeq(T^{-1}ET)$. Hence, a single representative element of each left coset of $T$ suffices. Topologically distinct graphs are generated by different left cosets of the row-stochastic transformation matrices. 
    The number of those cosets is the cardinality of the transformation matrix group divided by the cardinality of the relevant permutation subgroup.
\end{Rem}

\section{\texorpdfstring{Algebraic Details on Matrices over $\mathbb{F}_p$}{Algebraic Details on Matrices over F\_p}} \label{sec:math_foundation} \noindent
Motivated by Proposition~\ref{Prop:mat_sim_transformation}, this section gives mathematical details on square matrices over $\mathbb{F}_{p}$. The algebraic structure below is detailed not for its own sake, but to justify the correctness and complexity of the generating algorithms in Section~\ref{sec:algorithms}. We denote the set of all $N \times N$ square matrices over $\mathbb{F}_{p}$ by $M_{N,p}$. Its cardinality is obtained by freely selecting $N^2$ entries from $\mathbb{F}_{p}$, which results in $p^{N^2}$ total possibilities. 

\subsection{\texorpdfstring{Set of Invertible Matrices $GL_{N,p}$}{Set of Invertible Matrices GL(N,p)}} \label{subsec:invertible_matrices} \noindent
Invertible matrices $T \in\mathbb{F}_{p}^{N\times N}$, with $\det(T) \ne 0 \pmod{p}$, form the \textit{General Linear Group} over $\mathbb{F}_{p}$, $GL_{N,p}$. The cardinality of $GL_{N,p}$ is a standard result,~\cite[Theorem~10.6]{arapuraAlgebraNotes},
\begin{align}
    |GL_{N,p}| = \prod\nolimits_{i=0}^{N-1} (p^N - p^i).
\end{align}

\subsection{\texorpdfstring{Set of Row-stochastic Matrices $M_{N,p}^{RS}$}{Set of Row-stochastic Matrices}} \label{subsec:RS_matrices} \noindent
A matrix $T\in \mathbb{F}_{p}^{N \times N}$ is \emph{row-stochastic (RS) over} $\mathbb{F}_{p}$ if all its entries are in $\mathbb{F}_{p}$ and each of its rows sums to one in $\mathbb{F}_{p}$, \textit{i.e.}, $T\mathbf{1}_N = \mathbf{1}_N \pmod{p}$.

The set of all RS matrices over $\mathbb{F}_{p}$ is denoted by $M_{N,p}^{RS}$. The RS constraint imposes $N$ linearly independent equations on the $N^2$ entries of $T$, so $M_{N,p}^{RS}$ defines an \textit{affine} subspace of $M_{N,p}$. 

\noindent \textbf{The set $M_{N,p}^{RS}$ is closed under matrix multiplication.} The RS constraint makes $\mathbf{1}_N$ a right eigenvector of $T$ for $\lambda = 1$, so the product of two RS matrices is RS; $T_1 \mathbf{1}_N = \mathbf{1}_N$ and $T_2 \mathbf{1}_N = \mathbf{1}_N$ implies that $(T_1 T_2) \mathbf{1}_N = T_1 (T_2 \mathbf{1}_N) = \mathbf{1}_N$.

As each of the $N$ rows has $N-1$ free entries, the last being fixed by $T_{i,N} = 1-\sum\nolimits_{j=1}^{N-1} T_{i,j}$, the cardinality of $M_{N,p}^{RS}$ is
\begin{align}
    \label{eq:row_stoch:card}
    |M_{N,p}^{RS}| = p^{N(N-1)}.
\end{align}

\subsection{\texorpdfstring{Set of Row-stochastic Invertible Matrices $G_{N,p}^{RS}$}{Set of Row-stochastic Invertible Matrices}} \label{subsec:row_inv_matrices} \noindent
The set of all invertible row-stochastic matrices is denoted by $G_{N,p}^{RS}$. As the RS property is closed under multiplication and $GL_{N,p}$ is a group, $G_{N,p}^{RS}=M_{N,p}^{RS}\cap GL_{N,p}$ is a \textit{subgroup} of $GL_{N,p}$ whose members fix the vector $\mathbf{1}_N \in \mathbb{F}_p^N$. Consequently, $G_{N,p}^{RS}$ is isomorphic to the group of matrices fixing \textit{any} other nonzero vector in $\mathbb{F}_{p}^N$, \textit{e.g.}, the standard basis vector $\mathbf{e}_N = (0, \dots, 0, 1)^{\top}$, as shown in the following.

\begin{Lem}
    \label{Lem:Q_similarity}
    Let $T$ be a row-stochastic and invertible matrix, $T \in G_{N,p}^{RS}$. Let $Q$ be a matrix structured as
    \begin{align}
        \label{eq:Q}
        Q = \begin{bmatrix}
            \bar{Q}_{N-1} & \mathbf{1}_{N-1} \\
            \mathbf{0}_{N-1}^{\top} & 1
        \end{bmatrix},
    \end{align}
    where $\bar{Q}_{N-1}$ is invertible; $\bar{Q}_{N-1} \in GL_{N-1,p}$. Assume that $S \in GL_{N,p}$ fixes $\mathbf{e}_N$, \textit{i.e.}, $S \mathbf{e}_N = \mathbf{e}_N$. Then, $QSQ^{-1} \in G^{RS}_{N,p}$ and $QSQ^{-1}=T$ for $S = Q^{-1}TQ$. This matrix similarity relation establishes a bijection and an isomorphism.
\end{Lem}

\begin{proof}
    First note that by construction $Q\mathbf{e}_N = \mathbf{1}_N$. As $T \in G_{N,p}^{RS}$ is row-stochastic, $T\mathbf{1}_N = \mathbf{1}_N$, one has $T \mathbf{1}_N = Q\mathbf{e}_N$, where $\mathbf{1}_N$ on the right-hand side of the row-stochasticity relation is replaced by $Q\mathbf{e}_N$. As $S \mathbf{e}_N = \mathbf{e}_N$, one obtains $T \mathbf{1}_N = Q \mathbf{e}_N = QS \mathbf{e}_N$. Finally, note that $\mathbf{e}_N = Q^{-1} \mathbf{1}_N$, so $T \mathbf{1}_N = Q \mathbf{e}_N = QSQ^{-1} \mathbf{1}_N = \mathbf{1}_N$. Hence, $QSQ^{-1}\in G_{N,p}^{RS}$ and $T\sim S$ for $S = Q^{-1}TQ$.
\end{proof}

The importance of Lemma~\ref{Lem:Q_similarity} is that once the invertible matrix $S$ is generated to fix $\mathbf{e}_N$, \textit{i.e.}, $S \mathbf{e}_N = \mathbf{e}_N$, the resulting $T = QSQ^{-1}$ is both RS and invertible, $T \in G^{RS}_{N,p}$. The need to independently verify nonsingularity of $T$ is thus avoided. Such a matrix $S$ has the following structure,
\begin{align}
    \label{eq:S}
    S = \begin{bmatrix} 
        S_{N-1} & \mathbf{0}_{N-1} \\ 
        \mathbf{c}^{\top} & 1 
    \end{bmatrix}, 
\end{align}
where $S_{N-1} \in M_{N-1,p}$ is a square matrix and $\mathbf{c} \in \mathbb{F}_p^{N-1}$ is a free vector of length $N-1$. For $S$ to be in $GL_{N,p}$, $S_{N-1}$ must be invertible; \textit{i.e.}, $S_{N-1} \in GL_{N-1,p}$. The bottom row vector $\mathbf{c}^{\top}$ can be chosen arbitrarily, providing $p^{N-1}$ options. The structure of matrix $Q$ is given in~\eqref{eq:Q}. One can use any invertible matrix $\bar{Q}_{N-1} \in GL_{N-1,p}$ and an obvious choice is $I_{N-1}$. As a result, the cardinality of $G^{RS}_{N,p}$ is
\begin{align}
    \vert G^{RS}_{N,p}\vert=p^{N-1}\prod\nolimits_{i=0}^{N-2} (p^{N-1} - p^i).
\end{align}

\subsection{Set of Permutation Matrices} \label{subsec:permutation_matrices} \noindent
A permutation matrix $P \in \mathcal{P}_N$ is doubly-stochastic, binary and orthogonal; $P^{-1} = P^{\top},\: P^{\top}P = I_N$. Permutation matrices are always invertible, as $\det(P) = \pm1 \neq 0$ in any field. Based on these, one has the following nested group inclusions, $\mathcal{P}_N \subset G_{N,p}^{RS} \subset GL_{N,p}$. The cardinality of the permutation group $\mathcal{P}_N$ is given by a purely combinatorial count, independent of $p$, equalling $N!$. 

Following Remark~\ref{Rem:TP_row_stochastic}, among all row-stochastic invertible matrices, $G_{N,p}^{RS}$, it is useful to identify those from the same left coset with respect to the permutation subgroup $\mathcal{P}_N$.

\noindent \textbf{Permutation check.} The following result is useful for determining if one matrix is a column or row permutation of another.

\begin{Prop}
    Let $P \in \mathcal{P}_N$ be a permutation transformation matrix so that, applied to any column vector $x$, $Px$ permutes its components. Then, for a matrix $T$, $PT$ permutes rows of matrix $T$, and $TP^\top$ permutes columns of matrix $T$.
\end{Prop}

Hence, the conditions for two matrices, $T_1$ and $T_2$, to be column or row permutations of one another are $T_1 = T_2P$ or $T_1 = PT_2$, respectively. These imply
\begin{align}
    \label{eq:per_check}
    T_2^{-1} T_1 = P \quad \text{or} \quad T_1 T_2^{-1} = P, 
\end{align}
respectively. Checking conditions~\eqref{eq:per_check} requires a matrix inversion, $\mathcal{O}(N^3)$, which is computationally costly.

Instead, one can sort the columns of both matrices lexicographically. To sort the columns of a matrix $T \in \mathbb{F}_p^{N \times N}$ lexicographically, one employs a comparison-based sorting algorithm using a vector-wise comparison operator: given $[T]_{:,i},\: [T]_{:,j} \in \mathbb{F}_p^N$, define $[T]_{:,i} <_\mathrm{lex} [T]_{:,j}$ if at the first index $k \in \{1, \dots, N\}$ where $T_{ki} \neq T_{kj}$, the condition $T_{ki} < T_{kj}$ holds in the standard ordering of the integers $\{0, \dots, p-1\}$. If one matrix is a column-permuted version of the other, their lexicographically sorted versions are identical. This is given in Algorithm~\ref{Alg:col:check}. 

\begin{Rem}
    The computational complexity of Algorithm~\ref{Alg:col:check} is $\mathcal{O}(N^2 \log N)$. Each comparison between two columns requires at most $N$ scalar comparisons, resulting in a complexity of $\mathcal{O}(N)$ \textit{per} operation. As sorting $N$ elements requires $\mathcal{O}(N \log N)$ comparisons, the total computational complexity for the whole matrix is $\mathcal{O}(N^2 \log N)$. 
\end{Rem}

\begin{algorithm}[!t]
    \caption{Optimized Column-Permutation Equivalence Check}
    \begin{algorithmic}[1] 
    \label{Alg:col:check}
    \REQUIRE Matrices $T_{1}, T_{2} \in \mathbb{F}_p^{N \times N}$
    \STATE Initialize $\texttt{permutation\_matrices} \gets$ \textbf{True}.
    \STATE Sort columns of $T_1$ lexicographically to obtain $T_1^{\text{sorted}}$.
    \STATE Sort columns of $T_2$ lexicographically to obtain $T_2^{\text{sorted}}$.
    \FOR{$j = 1$ to $N$}
        \IF{$[T_1^{\text{sorted}}]_{:,j} \neq [T_2^{\text{sorted}}]_{:,j}$}
            \STATE \texttt{permutation\_matrices} $\gets$ \textbf{False}.
            \STATE \textbf{break}
        \ENDIF
    \ENDFOR
    \RETURN \texttt{permutation\_matrices}.
    \end{algorithmic}
\end{algorithm}

\section{\texorpdfstring{Algorithms to Generate the Matrix $T$}{Algorithms to Generate the Matrix T}} \label{sec:algorithms} \noindent
This section gives two algorithms to generate the invertible row-stochastic transformation matrix $T$ of Proposition~\ref{Prop:mat_sim_transformation}.

\subsection{Sampling and Rejection Algorithm} \label{subsec:SARA} \noindent
The simplest approach relies on generating a uniformly random matrix in $M_{N,p}^{RS}$ and checking the invertibility and permutation conditions. The \textit{Sampling and Rejection} method (SAR) is summarized in Algorithm~\ref{Alg:SAR}. Once a new matrix $T$ is generated by Algorithm~\ref{Alg:SAR}, one should check if $T$ is a column permutation of some previously generated matrix, or not, by Algorithm~\ref{Alg:col:check}, and exclude $T$ if it is.

\begin{Rem}
    The selection process involves $\mathcal{O}(N^2)$ random finite-field element selections. As the sum of independent, uniformly distributed random variables over a field $\mathbb{F}_{p}$ is itself uniformly distributed, the resulting $T$ is guaranteed to be uniformly distributed over the whole affine space $M_{N,p}^{RS}$,~\cite{pasqualetti2014consensus_networks}. One can determine if $T$ is nonsingular by calculating its determinant, or performing Gaussian elimination over $\mathbb{F}_{p}$, with a computational complexity of $\mathcal{O}(N^3)$. 
\end{Rem}

\begin{algorithm}[!t] 
    \caption{Sampling and Rejection}
    \begin{algorithmic}[1] 
        \label{Alg:SAR}
        \STATE \textbf{Sampling step:} 
        \STATE Sample $T_{i,j},\: i=1,...,N,\: j=1,...,N-1$ uniformly and independently from $\mathbb{F}_{p}$.
        \label{TF:s1}
        \STATE Set $T_{i, N} = 1 - \sum\nolimits_{j=1}^{N-1} T_{i, j} $.
        \STATE \textbf{Rejection step:}
        \IF{$\det(T) \ne 0$ and $T \notin \mathcal{P}_N$}
            \RETURN $T$.
        \ELSE
            \STATE Go to Step~\ref{TF:s1}.
        \ENDIF
    \end{algorithmic}
\end{algorithm}

\begin{Lem}
    \label{Lem:delta_RS_Np}
    The probability of successfully generating an invertible row-stochastic matrix $T$ which is \textit{not} a permutation matrix, $T \in G_{N,p}^{RS} \setminus \mathcal{P}_N$, by Algorithm~\ref{Alg:SAR}, is given by 
    \begin{align}
        \label{eq:delta}
        \delta_{N,p}^{RS} = \prod\nolimits_{i=1}^{N-1}(1-\frac{1}{p^i})-\frac{N!}{p^{N(N-1)}}.
    \end{align}
\end{Lem}

\begin{proof}
    The probability $\delta_{N,p}^{RS}$ precisely equals the ratio of the cardinality of $G_{N,p}^{RS} \setminus \mathcal{P}_N$ to the cardinality of the set $M_{N,p}^{RS}$ of all row-stochastic matrices, from which Algorithm~\ref{Alg:SAR} samples uniformly,
    \begin{equation*}
        \delta_{N,p}^{RS} = \frac{ p^{N-1} \prod\nolimits_{i=0}^{N-2} (p^{N-1} - p^i)-N!}{p^{N(N-1)}}.
    \end{equation*}
    Factoring $p^{N-1}$ out of each of the $N-1$ terms of the product cancels $p^{(N-1)^2}$ and yields~\eqref{eq:delta}.
\end{proof}

Based on~\eqref{eq:delta}, $\delta_{N,p}^{RS}$ is zero only for $N = 2,\: p = 2$ and positive otherwise. Moreover, $\delta_{N,p}^{RS}$ increases with the field characteristic $p$ for a fixed matrix dimension $N$, and
\begin{equation*}
    \lim_{N \rightarrow \infty} \delta_{N,p}^{RS} = \prod\nolimits_{i=1}^{\infty} \Big(1-\frac{1}{p^i}\Big) > 0, \qquad \lim_{p \rightarrow \infty} \delta_{N,p}^{RS} = 1.
\end{equation*}
Hence, the success probability stays bounded away from zero as $N$ grows, and for large $p$ almost every sampled $T$ is valid. Figure~\ref{fig:delta_TRS_NP} gives the 3D plot of the probability $\delta_{N,p}^{RS}$ based on the matrix dimension $N$ and the field characteristic $p$. 

\begin{figure}[htbp]
    \centerline{\includegraphics[trim=110 10 50 50, clip, width=0.465\textwidth]{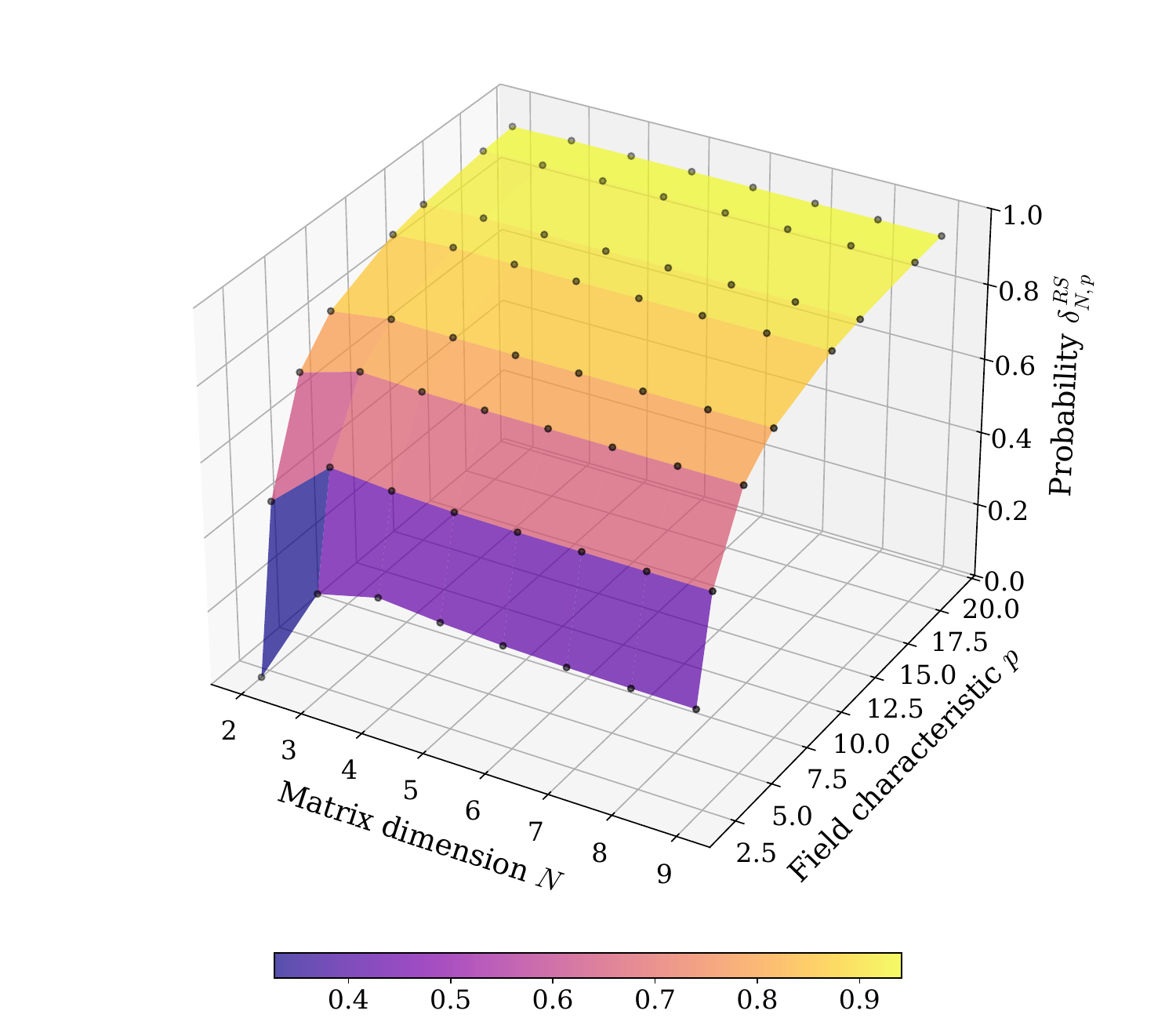}}
    \caption{The probability $\delta_{N,p}^{RS}$ based on the matrix dimension $N$ and the field characteristic $p$.}
    \label{fig:delta_TRS_NP}
\end{figure}

\subsection{Triangular Structure} \label{subsec:triangular_structure} \noindent
In this subsection, we restrict the matrix $T$ to be \textit{triangular}. The main motivation is that if diagonal entries of a triangular $T$ are nonzero, $T$ is invertible, and the off-diagonal elements can be randomly generated to render $T$ row-stochastic. As a result, one requires neither calculating the determinant nor checking row dependencies, making the algorithm extremely efficient, with computational complexity $\mathcal{O}(N^2)$. Let $U_{N,p}^{RS}$ denote the set of all \textit{upper} triangular, invertible, RS matrices that are \textit{not} permutation matrices. As the only triangular permutation matrix is the identity matrix, $U_{N,p}^{RS}$ excludes only $I_N$. Therefore, no column-permutation equivalence checks are required here. 

Algorithm~\ref{Alg:TF} summarizes generating the upper triangular matrices $T \in U_{N,p}^{RS}$. The algorithm for generating lower triangular matrices is similar and omitted for brevity. We collectively refer to algorithms generating upper or lower triangular matrices as \textit{Triangular Form} (TF) algorithms.

\begin{Lem}
    \label{Lem:U_Np:card}
    The number of upper triangular invertible row-stochastic non-identity matrices is given by 
    \begin{align}
        \label{eq:U_Np:card}
        |U_{N,p}^{RS}| = (p-1)^{N-1} p^{(N-1)(N-2)/2}-1.
    \end{align}
\end{Lem}

\begin{proof}
    A row-stochastic upper triangular matrix can be structured as follows: the last column is initially reserved and later selected to preserve the RS property; $N-1$ entries, (all diagonal elements except for the last one), are selected as nonzero. The last diagonal entry is set to $1$ to make $T$ row-stochastic. The remaining $(N-1)(N-2)/2$ off-diagonal entries are selected randomly. For a triangular matrix, the only possible permutation matrix is the identity matrix, which is excluded as trivial. As a result, the number of upper triangular invertible row-stochastic matrices is~\eqref{eq:U_Np:card}.
\end{proof}

\begin{algorithm}[H] 
    \caption{Upper Triangular Form}
    \begin{algorithmic}[1] 
        \label{Alg:TF}
        \STATE \textbf{Sampling step:} 
        \STATE Sample $T_{i,i} \in \mathbb{F}_p \setminus \{0\}$ for $i = 1, \dots, N-1$ uniformly.
        \STATE Sample $T_{i,j}$ for $i=1, \dots, N-2,\:j=i+1, \dots, N-1$ uniformly. 
        \label{TF:r1}
        \STATE Set the last column and last diagonal entry
        \begin{itemize}
            \item $T_{i,N} = 1 - \sum_{j=i}^{N-1} T_{i, j} $ for $i=1, \dots, N-1$,
            \item $T_{N,N} = 1$.
        \end{itemize}
        \STATE \textbf{Rejection step:}
        \IF{$T = I_N$}
        \STATE Go to Step~\ref{TF:r1}.
        \ELSE \RETURN $T$.
        \ENDIF
    \end{algorithmic}
\end{algorithm}

An analogous result for lower triangular invertible row-stochastic matrices can be established similarly. Lemma~\ref{Lem:U_Np:card} specifies that the number of upper, (and similarly lower), triangular matrices $T$ satisfying $T \in G_{N,p}^{RS},\: T\neq I_N$, increases with both the matrix dimension $N$ and the field characteristic $p$; \textit{i.e.}, one can find more matrices $T \in U_{N,p}^{RS}$ for larger $N$ and $p$, \textit{cf}. Figures~\ref{fig:carU_N}~and~\ref{fig:carU_p}. Note that, with $I_N$, the set of upper and the set of lower triangular invertible row-stochastic matrices constitute distinct subgroups of $G_{N,p}^{RS}$.

\begin{figure}[tbp]
    \centerline{\includegraphics[trim=10 13 10 10, clip, width=0.485\textwidth]{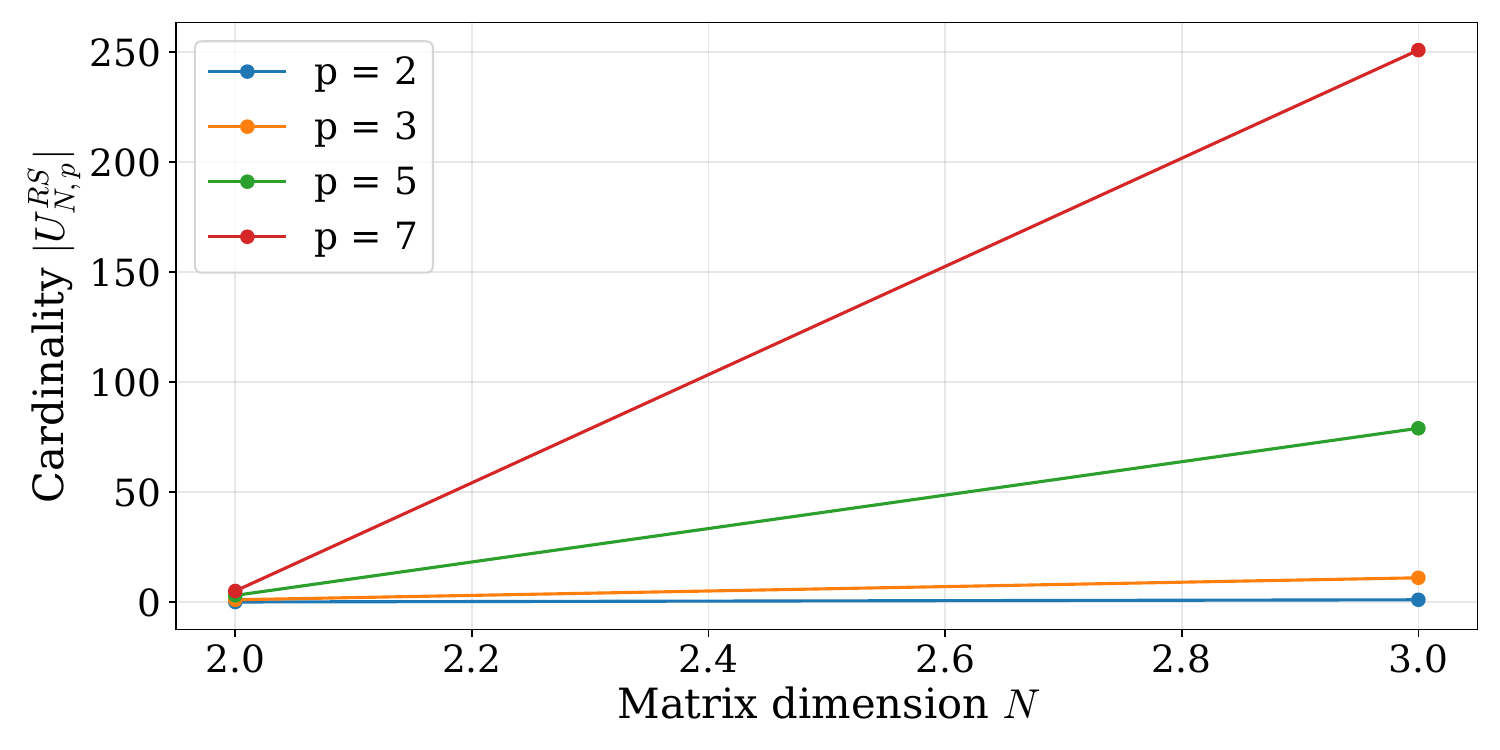}}
    \caption{The cardinality $|U_{N,p}^{RS}|$, based on the matrix dimension $N$.}
    \label{fig:carU_N}
\end{figure}

\begin{figure}[tbp]
    \centerline{\includegraphics[trim=10 13 10 10, clip, width=0.485\textwidth]{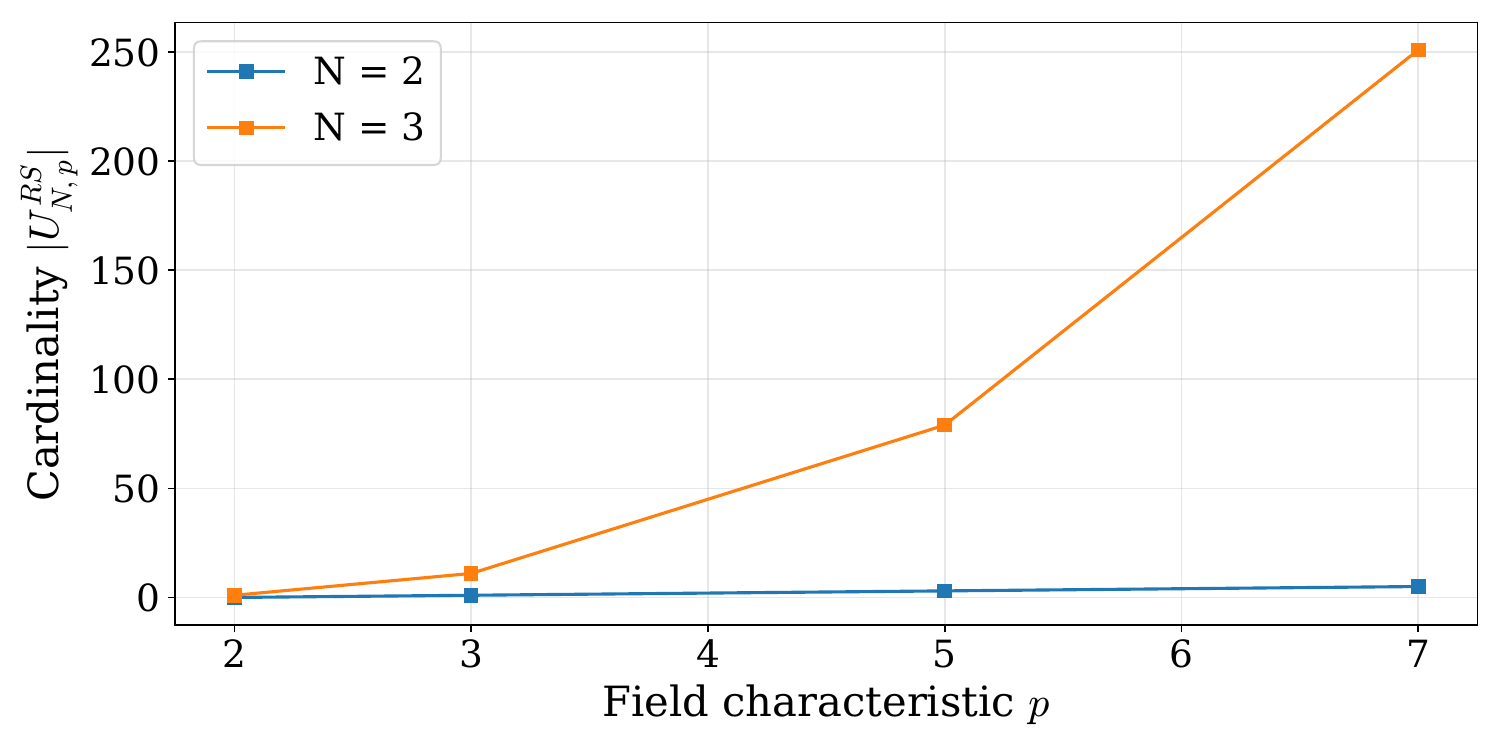}}
    \caption{The cardinality $|U_{N,p}^{RS}|$, based on the field characteristic $p$.}
    \label{fig:carU_p}
\end{figure}

\section{Simulation Results} \label{sec:sim_results} \noindent
This section provides a numerical validation of the proposed framework and algorithms using a small-scale example. Consider a network of $N = 2$ agents operating over $\mathbb{F}_3$. 

\subsection{Matrix Generation and Success Probability} \label{subsec:matrix_gen_den} \noindent
The space of all row-stochastic matrices $M_{2,3}^{RS}$ has a dimension of $N(N-1) = 2$, resulting in $|M_{2,3}^{RS}| = 3^2 = 9$ possible matrices. For a matrix
\begin{equation*}
    T = \begin{bmatrix} a & 1-a \\ c & 1-c \end{bmatrix} \in M_{2,3}^{RS} 
\end{equation*}
to be invertible, ($T \in G_{2,3}^{RS}$), we require $\det(T) = a - c \neq 0$. This condition is satisfied by six matrices. 

Excluding the two permutation matrices, \textit{i.e.}, $I_2$ and the swap matrix, 
\begin{equation*}
    \begin{bmatrix} 0 & 1 \\ 1 & 0 \end{bmatrix}, 
\end{equation*}
we obtain four valid $T$ matrices. The success probability for the SAR algorithm is thus $\delta_{2,3}^{RS} = 4/9$. Our simulation results confirm that the SAR algorithm successfully generates all four non-permutation matrices in $G_{2,3}^{RS}$. The TF algorithm successfully identifies the two triangular non-permutation matrices, \textit{i.e.}, 
\begin{equation*}
    T = \begin{bmatrix} 1 & 0 \\ 2 & 2 \end{bmatrix} \quad \text{and} \quad T = \begin{bmatrix} 2 & 2 \\ 0 & 1 \end{bmatrix}.
\end{equation*}

\subsection{Topological Admissibility} \label{subsec:topol_admissibility} \noindent
The four valid $T$ matrices in $G_{2,3}^{RS} \setminus \mathcal{P}_2$ form \textit{two} distinct cosets with respect to $\mathcal{P}_2$, each having two elements. By applying these to an admissible seed matrix
\begin{equation*}
    E_0 = \begin{bmatrix} 0 & 1 \\ 0 & 1 \end{bmatrix},
\end{equation*}
we synthesize all admissible communication graphs. As shown in our numerical tests, this yields topologically equivalent $E$ matrices, \textit{e.g.}, 
\begin{equation*}
    \begin{bmatrix} 1 & 0 \\ 1 & 0 \end{bmatrix}, 
\end{equation*}
and topologically distinct ones, \textit{e.g.}, 
\begin{equation*}
    \begin{bmatrix} 2 & 2 \\ 2 & 2 \end{bmatrix}. 
\end{equation*}
This confirms that the proposed decoupling allows for efficient identification of all the graphs admissible for finite-field consensus and synchronization. 

Furthermore, the initial admissible graph matrix 
\begin{equation*}
    E_0 = \begin{bmatrix} 0 & 1 \\ 0 & 1 \end{bmatrix} 
\end{equation*}
can be generated from the Jordan form 
\begin{equation*}
    J_0 = \begin{bmatrix} 0 & 0 \\ 0 & 1 \end{bmatrix}, 
\end{equation*}
which has the required spectrum, through a matrix similarity transformation with 
\begin{equation*}
    Q = \begin{bmatrix}1 & 1 \\ 0 & 1 \end{bmatrix}. 
\end{equation*}
As $J_0 \mathbf{e}_2 = \mathbf{e}_2$ and $Q \mathbf{e}_2 = \mathbf{1}_2$, (\textit{cf}. Lemma~\ref{Lem:Q_similarity}), one obtains a row-stochastic matrix 
\begin{equation*}
    E_0 = QJ_0Q^{-1} = \begin{bmatrix} 0 & 1 \\ 0 & 1 \end{bmatrix}.
\end{equation*}
All discussed algorithms are implemented in Python and made publicly available at \url{https://github.com/simonlehky/generating_T_modulo_p}.

The example is deliberately small to keep the relations between its elements transparent; the approach itself applies to any number of agents $N$ and field characteristic $p$. 

\section{Conclusion} \label{sec:conclusion} \noindent
This paper deals with finite-field consensus and synchronization in a unified way. Consensus and synchronization of finite-field systems are difficult precisely because finding all admissible graphs is NP-hard. Crucially, we propose efficient algorithms to find the admissible graphs. Our algorithms scale better than the combinatorially difficult comprehensive search. Numerical results demonstrate the effectiveness of this approach. One compelling future direction is to extend these results to time-varying or switching graph topologies.

\bibliographystyle{IEEEtran}
\bibliography{biblio}

\end{document}